\documentclass[10pt,A4paper,conference]{IEEEtran}
\usepackage{amsmath}
\usepackage{amssymb}
\usepackage{upref}
\usepackage{amsfonts}
\usepackage{graphicx}
\usepackage{verbatim}

\newcommand{\field}[1]{\mathbb{#1}}

\newcommand{\Z}{\field{Z}}

\newcommand{\R}{\field{R}}

\newcommand{\cA}{{\cal A}}
\newcommand{\cB}{{\cal B}}

\newcommand{\cD}{{\cal D}}
\newcommand{\cG}{{\cal G}}

\newcommand{\cP}{{\cal P}}

\newcommand{\cS}{{\cal S}}

\newcommand{\cT}{{\cal T}}

\newcommand{\sP}{\cP}
\newcommand{\sG}{\cG}

\newcommand{\Gr}{\smash{{\sG\kern-1.5pt}_q\kern-0.5pt(n,k)}}
\newcommand{\Grtwo}{\smash{{\sG\kern-1.5pt}_2\kern-0.5pt(n,k)}}
\newcommand{\Gkone}{\smash{{\sG\kern-1.5pt}_q\kern-0.5pt(n,k_1)}}
\newcommand{\Gktwo}{\smash{{\sG\kern-1.5pt}_q\kern-0.5pt(n,k_2)}}
\newcommand{\Ps}{\smash{{\sP\kern-2.0pt}_q\kern-0.5pt(n)}}

\newcommand{\deff}{\mbox{$\stackrel{\rm def}{=}$}}

\newtheorem{theorem}{Theorem}

\newtheorem{lemma}{Lemma}
\newtheorem{remark}{Remark}

\newtheorem{example}{Example}

\begin{document}

\title{Sidon Sequences and Doubly Periodic\\ Two-Dimensional Synchronization Patterns}

\author{\authorblockN{Tuvi Etzion}
\authorblockA{Department of Computer Science\\
Technion-Israel Institute of Technology\\
Haifa 32000, Israel \\
Email: etzion@cs.technion.ac.il}}

\maketitle
\begin{abstract}
Sidon sequences and their generalizations have found during the
years and especially recently various applications in coding
theory. One of the most important applications of these sequences
is in the connection of synchronization patterns. A few
constructions of two-dimensional synchronization patterns are
based on these sequences. In this paper we present sufficient
conditions that a two-dimensional synchronization pattern can be
transformed into a Sidon sequence. We also present a new
construction for Sidon sequences over an alphabet of size
$q(q-1)$, where $q$ is a power of a prime.
\end{abstract}

\section{Introduction}
\label{sec:introduction}

Let $\cA$ be an abelian group and let $\cD = \{ a_1 , a_2 , \ldots
, a_m \} \subseteq \cA$ be a subset of $m$ distinct elements of
$\cA$. $\cD$ is a \emph{Sidon sequence} (or a $B_2$-sequence) over
$\cA$ if all the sums $a_{i_1} + a_{i_2}$ with $1 \leq i_1 \leq
i_2 \leq m$ are distinct (if $i_1 < i_2$ in the definition the
sequence is called a \emph{weak Sidon sequences}). Sidon sequences
have found many applications in coding and communication. For
example, weak Sidon sequences are used for construction of
constant weight codes with minimum Hamming distance 6~\cite{BSSS},
and constructions of location-correcting codes~\cite{RoSe96}.
Sidon sequences were used in constructions of two-dimensional
synchronization patterns~\cite{BEMP1,Etz11}. There is a
generalization to $B_h$ sequences (all sums of $h$ elements are
distinct) and they applied for example in multihop paths related
to wireless sensor networds~\cite{BEMP2} and error-correcting
codes for rank modulation~\cite{BaMa10}. A comprehensive survey on
$B_2$-sequences and their generalizations was given by
O'Bryant~\cite{Bry04}. Even so in a Sidon sequence all sums of
pairs of elements from $\cD$ (not necessarily distinct elements)
are distinct there is a trivial connection to a set in which all
differences of ordered pairs of elements are distinct.

\begin{theorem}
\label{thm:diff} A subset $\cD = \{ a_1 , a_2 , \ldots , a_m \}
\subseteq \cA$ is a Sidon sequence over $\cA$ if and only if all
the differences $a_{i_1} - a_{i_2}$ with $1 \leq i_1 \neq i_2 \leq
m$ are distinct in $\cA$.
\end{theorem}

A Sidon sequence with $m$ elements over an abelian group with $n$
elements is called \emph{optimal} if all Sidon sequences over an
abelian group with $n$ elements have at most $m$ elements. In view
of Theorem~\ref{thm:diff} bounds on the size of a Sidon sequence
(on the number of elements $m$) can be derived by considering
difference and not sums. This is important since the number of
distinct sums is $\binom{m}{2} +m = \frac{m^2 +m}{2}$ while the
number of distinct differences is considerably higher, $m(m-1) =
m^2 -m$. This yields a better upper bound on $m$. A Sidon sequence
$\cD$ is a set of $m$ elements. If the abelian group is $\Z_n$
then $\cD$ can be represented as a binary cyclic sequence $s=[s_0
s_1 , \cdots , s_{n-1} ]$, where $s_i =1$ if $i \in \cD$.

One-dimensional synchronization patterns were first introduced by
Babcock in connection with radio interference~\cite{Bab}. Other
applications are discussed in details in~\cite{BlGo77} and some
more are given in~\cite{ASU,LaSa88}. The two-dimensional
applications and related structures were first introduced
in~\cite{GoTa82} and discussed in many papers,
e.g.~\cite{GoTa84,Rob85,Games87,BlTi88,Rob97}. Recent
new application in keys predistribution for wireless sensor
networks~\cite{BEMP} led to new related two-dimensional problems
concerning these patterns~\cite{BEMP1,BEMP2}. Difference pattern and
Sidon sequences have an important role in the construction of synchronization patterns.

Some of the applications of Sidon sequence is due to the
difference properties implied by Theorem~\ref{thm:diff}. This
property is also the basis of the applications to two-dimensional
synchronization patterns. There are various papers,
e.g.~\cite{BEMP1,Etz11,Rob97} in which an one-dimensional sequence
(as a Sidon sequence or a ruler) is transformed into a
two-dimensional synchronization pattern. The main goal of this
paper is to establish the inverse transformation, in which a
two-dimensional synchronization pattern is transformed into a
Sidon sequence, which is a one-dimensional sequence.

The rest of this paper is organized as follows. In
Section~\ref{sec:period} we define what is a period in a
two-dimensional array and as a result we obtain a definition for a
cyclic two-dimensional array. In Section~\ref{sec:DDCs} we discuss
various types of two-dimensional synchronization patterns. In
particular we discuss periodic two-dimensional synchronization
patterns. In Section~\ref{sec:folding} we present two operations,
namely, folding and unfolding. Folding generates a two-dimensional
array from an one-dimensional sequence. Unfolding is the inverse
operation and it generates an one-dimensional sequence from a
two-dimensional array. In particular we will prove that these
operations relate periodic sequences to periodic two-dimensional
arrays and vice-versa. Moreover, they relate an one-dimensional
synchronization sequence to a two-dimensional synchronization
array and vice-versa, if the two-dimensional array is periodic. As
a consequence we obtain the main result of the paper that
two-dimensional periodic synchronization arrays which can be
unfolded are equivalent to Sidon sequences over $\Z_n$, where $n$
is the size of one period in the array. In Section~\ref{sec:new}
we present a construction of optimal Sidon sequences with $q-1$
elements over a group with $q(q-1)$ elements, where $q$ is a power
of a prime. This generalizes a similar result where $q$ is a
prime. Section~\ref{sec:conclude} contains conclusions and
problems for further research.

\section{Periodicity of Two-Dimensional Arrays}
\label{sec:period}

\subsection{periodic sequences and arrays}

It is very simple to define the periodicity for one-dimensional
sequences. An infinite sequence $S= \ldots s_{-1} , s_0 , s_1 ,
s_2 , \ldots$ is periodic if there exists an integer $\pi$ such
that $s_{i+\pi} = s_i$ for each $i \in \Z$. If $\pi$ is the
smallest integer for which the sequence has this property then we
say that $\pi$ is the \emph{period} of the sequence and write the
sequence as $[s_0 , s_1 ,\ldots , s_{\pi -1}]$, and say that the
sequences is a \emph{cyclic sequence} or a \emph{cycle}. It is
well known that
\begin{theorem}
\label{thm:period_one} If $\pi$ is the period of a sequence $S$,
and there exists an integer $\rho$ such that $s_{i+\rho} = s_i$
for each $i \in \Z$, then $\pi$ divides $\rho$.
\end{theorem}

Usually, an infinite two-dimensional array $\cA$ is said to be
doubly periodic if there exists two integers $\kappa$ and $\eta$
such that for each $i,j \in \Z$ we have $\cA(i+\kappa
,j)=\cA(i,j+\eta)=\cA(i,j)$. But, it appears that this definition
is too restricted. A generalized definition, which give more
information, is as follows. An infinite two-dimensional array $A$
is \emph{doubly periodic} if there exists two linearly independent
integer vectors $(\pi_1 , \pi_2 )$ and $(\xi_1 , \xi_2 )$ such
that each $i,j \in \Z$ satisfy $\cA(i+\pi_1 ,j+\pi_2 )=\cA(i+\xi_1
,j+\xi_2 )=\cA(i,j)$. How we can define the smallest vectors with
this property? What is the period of the array? and what is a
cyclic two-dimensional array? These questions will be answered
after two necessary definitions, of tiling and lattices, will be
presented.

\subsection{Tiling}

Tiling is one of the most basic concepts in combinatorics. We say
that a two-dimensional shape $\cS$ tiles the two-dimensional
square grid $\Z^2$ if disjoint copies of $\cS$ cover $\Z^2$. This
cover of $\Z^2$ with disjoint copies of $\cS$ is called a {\it
tiling} of $\Z^2$ with $\cS$. For each shape $\cS$, in the tiling,
we distinguish one of the points of $\cS$ to be the {\it center}
of $\cS$. Each copy of $\cS$ in a tiling has the center in the
same related point. The set $\cT$ of centers in a tiling defines
the tiling, and hence the tiling is denoted by the pair
$(\cT,\cS)$. Given a tiling $(\cT,\cS)$ and a grid point
$(i_1,i_2)$ we denote by $c(i_1,i_2)$ the center of the copy of
$\cS$, $\cS'$, for which $(i_1,i_2) \in \cS'$. We will also assume
that the origin is a center of a copy of $\cS$. The first lemma
given in~\cite{Etz11} can be easily verified.
\begin{lemma}
\label{lem:center} For a given tiling $(\cT,\cS)$ and a point
$(i_1,i_2)$ the point
$(i_1,i_2)-c(i_1,i_2)$ belongs to the shape
$\cS$ whose center is in the origin.
\end{lemma}

\subsection{Lattices and Lattice Tiling}

One of the most common types of tiling is a {\it lattice tiling}.
A two-dimensional {\it lattice} $\Lambda$ is a discrete, additive subgroup of the
real two-dimensional space $\R^2$. W.l.o.g., we can assume that

\begin{equation}
\label{eq:lattice_def} \Lambda = \{ u_1 v_1 + u_2v_2 ~:~ u_1, u_2 \in \Z \}
\end{equation}
where $v_1, ~v_2$ are two linearly independent
vectors in $\R^2$. A lattice $\Lambda$ defined by
(\ref{eq:lattice_def}) is a sublattice of $\Z^2$ if and only if
$\{ v_1,v_2\} \subset \Z^2$. We will be interested
solely in sublattices of $\Z^2$. The vectors $v_1,v_2$
are called {\it basis} for $\Lambda \subseteq \Z^2$, and the $2
\times 2$ matrix
$$
{\bf G}=\left[\begin{array}{cc}
v_{11} & v_{12} \\
v_{21} & v_{22} \end{array}\right]
$$
having these vectors as its rows is said to be the {\it generator
matrix} for $\Lambda$. Note, that it is always possible to use a
generator matrix ${\bf G}$ in which all the four entries are
nonzeroes. It is also always possible to have in ${\bf G}$ exactly
one {\it zero} entry.

The {\it volume} of a lattice $\Lambda$, denoted $V( \Lambda )$,
is inversely proportional to the number of lattice points per unit
volume. More precisely, $V( \Lambda )$ may be defined as the
volume of the {\it fundamental parallelogram} $\Pi(\Lambda)$ in
$\R^2$, which is given by
$$
\Pi(\Lambda) \deff\ \{ \xi_1 v_1  + \xi_2 v_2 ~:~ 0 \leq \xi_i < 1, ~ ,i=1,2 \}
$$
There is a simple expression for the volume of $\Lambda$, namely,
$V(\Lambda)=| \det {\bf G} |$.

We say that $\Lambda$ induces a {\it lattice tiling} of $\cS$ if
the lattice points can be taken as the set $\cT$ to form a tiling
$(\cT,\cS)$.

\subsection{Cyclic Arrays and Periods}

We are now in a position to define the period of a doubly periodic
array and to define cyclic two-dimensional arrays. Let $\cA$ be a
doubly periodic two-dimensional array. Let $(\pi_1 , \pi_2 )$ and
$(\xi_1 , \xi_2 )$ two linearly independent vectors such that each
$i,j \in \Z$ satisfy $\cA(i+\pi_1 ,j+\pi_2 )=\cA(i+\xi_1 ,j+\xi_2
)=\cA(i,j)$. Let $s$ be the volume of the lattice formed from
$(\pi_1 , \pi_2 )$ and $(\xi_1 , \xi_2 )$. Let $(\pi_3 , \pi_4 )$
and $(\xi_3 , \xi_4 )$ be two linearly independent vectors for
which, each $i,j \in \Z$ satisfy $\cA(i+\pi_3 ,j+\pi_4
)=\cA(i+\xi_3 ,j+\xi_4 )=\cA(i,j)$. Let $s'$ be the volume of the
lattice formed from $(\pi_3 , \pi_4 )$ and $(\xi_3 , \xi_4 )$. If
$s' \geq s$ for each such pair of linearly independent vectors
then we say that $\{ (\pi_1 , \pi_2 ) , (\xi_1 , \xi_2 ) \}$ is
the \emph{period} of $\cA$ and $s$ is the \emph{volume} of $\cA$.
The period in the one-dimensional case has the role of the period
and the volume in the two-dimensional case.

Clearly, the period of a two-dimensional array is not unique. The
volume of the array is unique and can be calculated
from the given period. We have a
theorem in the two-dimensional case which is akin to Theorem~\ref{thm:period_one}.

\begin{theorem}
Let $\cA$ be a doubly periodic array with period $\{ (\pi_1 ,
\pi_2 ) , (\xi_1 , \xi_2 ) \}$. Let $(\pi_3 , \pi_4 )$ and $(\xi_3
, \xi_4 )$ be two linearly independent vectors for which, each
$i,j \in \Z$ satisfy $\cA(i+\pi_3 ,j+\pi_4 )=\cA(i+\xi_3 ,j+\xi_4
)=\cA(i,j)$. Let $s'$ be the volume of the lattice formed from
$(\pi_3 , \pi_4 )$ and $(\xi_3 , \xi_4 )$. If $s$ is the volume of
the lattice formed from $(\pi_1 , \pi_2 )$ and $(\xi_1 , \xi_2 )$
then $s$ divides $s'$.
\end{theorem}

A shape $\cS$ will be called \emph{cyclic} if there is a lattice
tiling $\Lambda$ for $\cS$. In a cyclic sequence the order of the
elements, in the sequence, is obvious. It is less obvious for a
two-dimensional shape. We will discuss this order in
Section~\ref{sec:folding}.

\section{Two-Dimensional Synchronization Arrays}
\label{sec:DDCs}

Several types of two-dimensional synchronization arrays are
defined in the literature. We start we a general definition which
was given in~\cite{BEMP1,BEMP2}. Let $\cS$ be a given shape, on
the square grid, with $m$ dots on grid points. $\cS$ is called a
\emph{distinct difference configuration} (DDC) if the
$\binom{m}{2}$ lines connecting dots are distinct either in their
length or in their slope. Several types of DDCs were defined in
the literature. The main focus of research which was done on this
topic is related to Costas arrays. A {\it Costas array} is an $m
\times m$ permutation array having exactly one dot in each row and
each column. Some results on Costas arrays are given
in~\cite{GoTa82,GoTa84,Gol84,Dra06,GoGo07}.

We now present a definition for a doubly periodic DDC. A
\emph{doubly periodic $\cS$-DDC} is a doubly periodic
two-dimensional array $\cA$ with period $\{ (\pi_1 , \pi_2 ) ,
(\xi_1 , \xi_2 ) \}$ such that the following three properties are
satisfied.

\begin{itemize}
\item The lattice formed by $(\pi_1 , \pi_2 )$ and $(\xi_1 , \xi_2
)$ is a lattice tiling for $\cS$.

\item Each copy of $\cS$ on the two-dimensional arrays $\cA$ is a
DDC.

\item In each two copies of $\cS$ in the tiling, the positions of
the dots are the same.
\end{itemize}

Doubly periodic DDCs and in particular Costas array were
considered in the past, e. g.~\cite{Etz11,MGC97,MoGo06}. There are
two essential constructions for Costas arrays, both of them form
doubly periodic DDCs. The first construction is due to Welch and
the second Construction is due to Golomb (with a variant of
Lempel)~\cite{GoTa82,GoTa84,Gol84}. We will present both of them
in their doubly periodic version.

\noindent {\bf The periodic Welch Construction:}

Let $\alpha$ be a primitive root modulo a prime $p$ and let $\cA$
be the square grid. For any integers $i$ and $j$, there is a dot
in $\cA(i,j)$ if and only if $\alpha^i \equiv j \bmod p$.

\begin{theorem}
Let $\cA$ be the array of dots from the Periodic Welch
Construction. Then $\cA$ is a doubly periodic $\cS$-DDC with
period $\{ (0,p),(p-1,0) \}$ and $\cS$ is a $p \times (p-1)$
rectangle.
\end{theorem}

\vspace{1mm} \noindent {\bf The periodic Golomb Construction:}

Let $\alpha$ and $\beta$ be two primitive elements in GF($q$), where
$q$ is a prime power. For any integers $i$ and $j$, there is a dot in
$\cA(i,j)$ if and only if $\alpha^i + \beta^j =1$.

\begin{theorem}
\label{thm:periodicGolomb} Let $\cA$ be the array of dots from the
Periodic Golomb Construction. Then $\cA$ is a doubly periodic
$\cS$-DDC with period $\{ (0,q-1),(q-1,0) \}$ and $\cS$ is a
$(q-1) \times (q-1)$ square.
\end{theorem}

There are many important questions concerning Costas arrays. A few
of them are related to the periodicity of the arrays. In
particular we have the following two question:

\begin{enumerate}
\item Is the Welch construction generates all singly periodic
Costas arrays? where a singly periodic Costas array of order $n$
is an $n \times \infty$ array in which each $n \times n$ sub-array
is a Costas array. Welch Construction has this property for
$n=p-1$.

\item Are there more constructions for Costas arrays with
periodicity property?
\end{enumerate}

The conjecture is NO for both questions. Some evidence that this
conjecture is true is given in~\cite{EGT89}. In fact it should be
said that is very likely that most if not all Costas arrays are
known since they derived from the known
constructions~\cite{DIR10}. In what follows we will throw more
evidence for the difficulty to produce new doubly periodic
$\cS$-DDCs with many dots, different from those constructed by
folding~\cite{Etz11}.

Costas arrays are only one family of DDCs, and doubly periodic
DDCs. Two other families which were considered in the literature,
are the sonar sequences~\cite{GoTa82,Games87,MGC97,EGRT92}, and
the Golomb rectangles~\cite{Rob85,Rob97}.

\section{The Folding and Unfolding Methods}
\label{sec:folding}

This section is devoted to with a transformation of a periodic
sequence into a doubly periodic array and a transformation of a
doubly periodic array into a periodic sequence. The two
transformations will be called folding and unfolding,
respectively, and as one might expect these two transformations
are inverse of each other. As a consequence of the definition of
folding we will be able to define the order of the elements in a
cyclic shape. We will give the known theorems on the necessary and
sufficient conditions that a folding exists. Based on these
results we will define the inverse operation of unfolding. This
will lead to the main theorem which will state when a doubly
periodic two-dimensional DDC (or a cyclic DDC) is unfolded into a
Sidon sequences.

The definition of folding involves a lattice tiling $(\cT,\cS)$,
where $\cS$ is the shape on which the folding is performed. A {\it
direction} is a nonzero integer vector $(d_1,d_2)$, where $d_1 ,
d_2 \in \Z$.

Let $\cS$ be a two-dimensional shape and let $\delta=(d_1,d_2)$ be
a direction. Let $\Lambda$ be a lattice tiling for a shape $\cS$,
and let $\cS_1$ be the copy of $\cS$, in the related tiling, which
includes the origin. We define recursively a \emph{folded-row}
starting in the origin. If the point $(i_1,i_2)$ is the current
point of $\cS_1$ in the folded-row, then the next point on its
folded-row is defined as follows:
\begin{itemize}
\item If the point $(i_1+d_1,i_2+d_2)$ is in
$\cS_1$ then it is the next point on the folded-row.

\item If the point $(i_1+d_1,i_2+d_2)$ is in $\cS_2 \neq \cS_1$
whose center is in the point $(c_1,c_2)$ then
$(i_1+d_1-c_1,i_2+d_2-c_2)$ is the next point on the folded-row
(by Lemma~\ref{lem:center} this point is on $\cS_1$).
\end{itemize}

The definition of folding is based on a lattice $\Lambda$, a
shape $\cS$, and a direction $\delta$. The triple
$(\Lambda,\cS,\delta)$ defines a folding if the definition yields
a folded-row which includes all the elements of $\cS$. It appears
that only $\Lambda$ and $\delta$ determines whether the triple
$(\Lambda,\cS,\delta)$ defines a folding. The role of $\cS$ is
only in the order of the elements in the folded-row; and of course
$\Lambda$ must define a lattice tiling for $\cS$.

The first two lemmas proved in~\cite{Etz11} are an immediate
consequence of the definitions and provide us concise conditions
whether the triple $(\Lambda,\cS,\delta)$ defines a folding.

\begin{lemma}
Let $(\cT,\cS)$ be a lattice tiling defined by the two-dimensional
lattice $\Lambda$ and let $\delta = (d_1,d_2)$ be a direction.
$(\Lambda,\cS,\delta)$ defines a folding if and only if the set
$\{ (i \cdot d_1,i \cdot d_2)-c(i \cdot d_1,i \cdot d_2) ~:~ 0
\leq i < | \cS | \}$ contains $| \cS |$ distinct elements.
\end{lemma}

\begin{lemma}
Let $(\cT,\cS)$ be a lattice tiling defined by the two-dimensional
lattice $\Lambda$ and let $\delta = (d_1,d_2)$ be a direction.
$(\Lambda,\cS,\delta)$ defines a folding if and only if $(|\cS|
\cdot d_1,\cS| \cdot d_2)-c(|\cS| \cdot d_1,|\cS| \cdot
d_2)=(0,0)$ and for each $i$, $0 < i < |\cS|$ we have $(i \cdot
d_1,i \cdot d_2)-c(i \cdot d_1,i \cdot d_2) \neq (0,0)$.
\end{lemma}

The next theorem determine precisely when the triple
$(\Lambda,\cS,\delta)$ defines a folding.

\begin{theorem}
\label{thm:new_cond_fold2D} Let $\Lambda$ be a lattice whose
generator matrix is given by
$$
{\bf G}=\left[\begin{array}{cc}
v_{11} & v_{12} \\
v_{21} & v_{22}
\end{array}\right]~,
$$
where all the entries of ${\bf G}$ are nonzeroes.
Let $d_1$ and $d_2$ be two positive integers and $\tau =
\text{g.c.d.}(d_1 , d_2)$. If $\Lambda$ defines a lattice tiling
for the shape $\cS$ then the triple $(\Lambda,\cS,\delta)$ defines
a folding

\begin{itemize}
\item with the direction $\delta =(+d_1,+d_2)$ if and only if
$\text{g.c.d.}(\frac{d_1 v_{22}-d_2 v_{21}}{\tau},\frac{d_2
v_{11}-d_1 v_{12}}{\tau})=1$ and $\text{g.c.d.}(\tau , | \cS
|)=1$;

\item with the direction $\delta =(+d_1,-d_2)$ if and only if
$\text{g.c.d.}(\frac{d_1 v_{22}+d_2 v_{21}}{\tau},\frac{d_2
v_{11}+d_1 v_{12}}{\tau})=1$ and $\text{g.c.d.}(\tau , | \cS
|)=1$;

\item with the direction $\delta =(+d_1,0)$ if and only if
$\text{g.c.d.}(v_{12},v_{22})=1$ and $\text{g.c.d.}(d_1 , | \cS
|)=1$;

\item with the direction $\delta =(0,+d_2)$ if and only if
$\text{g.c.d.}(v_{11},v_{21})=1$ and $\text{g.c.d.}(d_2 , | \cS
|)=1$.
\end{itemize}
\end{theorem}

A direction $\delta$ for which $(\Lambda,\cS,\delta)$ defines a
folding also defines the order of the elements in a cyclic shape
$\cS$. This order is exactly the order of the elements in the
folded-row. It is easy to verify that only $|\cS|-1$ directions
should be considered for the existence of a folding. The order of
elements on $\cS$ is clearly not unique as it was proved
in~\cite{Etz11} that if one direction defines a folding then $\phi
(| \cS |)$ directions define a folding (and they come in pairs of
reverse order), where $\phi$ is the Euler totient function.

The unfolding operation is defined directly from the folding
operation. Let $\Lambda$ be a lattice tiling for a two-dimensional
shape $\cS$ and a let $\delta$ be a direction, for which
$(\Lambda,\cS,\delta)$ defines a folding. Then the folded-row is
the \emph{unfolded sequence} generated from the shape $\cS$. In
the folding, the folded-row indicates to which position of the
array, each element of a given one-dimensional sequence will be
assigned. In the unfolding, the folded-row is actually an
unfolded-row and it indicates to which position of the sequence,
each element of the array is assigned. These definitions are
completely natural and there is no surprise. What is more
interesting is the following theorem which connects
two-dimensional doubly periodic $\cS$-DDCs with Sidon sequences.

\begin{theorem}
\label{thm:2DtoSidon}
Let $\cA$ be a two-dimensional doubly periodic $\cS$-DDC with
period $\{ (\pi_1 , \pi_2 ) , (\xi_1 , \xi_2 ) \}$. Let $\Lambda$
be the lattice tiling of $\cS$ formed from $(\pi_1 , \pi_2 )$ and
$(\xi_1 , \xi_2 )$. If $\delta$ is a direction for which
$(\Lambda,\cS,\delta)$ defines a folding then the folded-row
generated by the unfolding of $\cS$ is a Sidon Sequence.
\end{theorem}
\begin{proof}
We will give a sketch of the proof. A proof with all the details
will appear in the full version of this work. We assign colors to
the points of the square grid as follows. The points of the shape
$\cS$ whose center is in the origin (say $\cS_0$) are assigned
colors by the order of the folded-row, where the origin is
assigned with a \emph{zero}. Each other copy of $\cS$ in the
tiling is assigned with the same colors as $\cS_0$ in the same
related positions.

Assume the contrary, that the folded-row is not a Sidon sequence.
It follows that there exists two distinct pairs of integers $(i_1
, i_2 )$ and $(i_3,i_4)$, where $i_1$, $i_2$, $i_3$, and $i_4$,
are positions with dots on the folded-row, such that $i_4 - i_3
\equiv i_2 - i_1 ~ (mod~n)$, where $n=|\cS|$. Let $\cS_j$, $1 \leq
j \leq 4$, a copy of $\cS$ on the grid in which a point colored
with $i_j$ is the center. It can be shown that in one of these
four copies the two pair of lines which connects the points
colored with $i_1$ and $i_2$ ad those colored with $i_3$ and $i_4$
are equal in length and slope. A contradiction to the assumption
that $\cA$ is
 a two-dimensional doubly periodic $\cS$-DDC.
\end{proof}

In~\cite{Etz11} the following theorem was proved.
\begin{theorem}
\label{thm:Sidonto2D} Let $\Lambda$ be a lattice tiling for a
two-dimensional shape $\cS$, $n=|\cS|$, and let $\delta$ be a
direction. Let $\cB$ be a Sidon sequence with $m$ elements over
$\Z_n$. If $(\Lambda,\cS,\delta)$ defines a folding then there
exists a two-dimensional doubly periodic $\cS$-DDC $\cA$ with $m$
dots in each copy of $\cS$ of $\cA$.
\end{theorem}

In view of Theorems~\ref{thm:2DtoSidon} and~\ref{thm:Sidonto2D} it
is tempting to prove that "a Sidon sequence over $\Z_n$ with $m$
elements exists if and only if a two-dimensional doubly periodic
$\cS$-DDC with $m$ dots in each copy of $\cS$ exists". But, this
claim is not correct. As we will conclude from the following
discussion, which applies Theorem~\ref{thm:2DtoSidon} on the known
doubly periodic constructions for Costas Arrays.

\begin{example} The following $7 \times 6$ array was obtained by
the periodic Welch Construction for $p=7$ and the primitive root 3
modulo 7.
$$
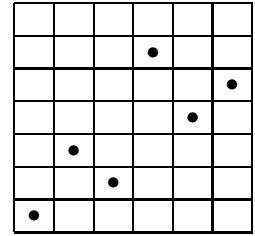
\begin{array}{|c|c|c|c|c|c|} \hline
&&&&&\\
\hline
&&&\bullet&&\\
\hline
&&&&&\bullet\\
\hline
&&&&\bullet&\\
\hline
&\bullet&&&&\\
\hline
&&\bullet&&&\\
\hline
\bullet&&&&&\\
\hline
\end{array}
$$
By using unfolding with direction $(1,1)$, where the lower left
dot is taken on the origin we obtain the Sidon sequence $\{
0,8,10,11,33,37\}$ modulo 42.
\end{example}
\begin{remark}
A construction such as the periodic Golomb Construction cannot
produce any Sidon sequence since the shape $\cS$ is a square and
the is no direction which defines a folding when $\cS$ is a
square.
\end{remark}

\section{New Optimal Sidon Sequences}
\label{sec:new}

The two celebrating constructions of optimal Sidon sequences are
the ones of Singer~\cite{Sin38} and Bose~\cite{Bose42}. Let $q$ be
a power of a prime number. Singer's construction, which is based
on projective planes, produces a Sidon sequence with $q+1$
elements over $\Z_{q^2 +q+1}$. Bose's construction, which is based
on affine planes, produces a Sidon set with $q$ elements over
$\Z_{q^2-1}$. The construction of Ruzsa~\cite{Ruz93} generates
optimal Sidon sequences with $p-1$ elements taken modulo $p^2 -p$,
where $p$ is a prime number. In this section we generalize this
construction to obtain a Sidon sequence with $q-1$ elements taken
over $(q-1) \times$GF($q$), where $q$ is any power of a prime.
Given a power of a prime $q$ and a primitive element $\alpha$ in
GF($q$), we construct the set $A_{q,\alpha}$ defined by
$$
A_{q,\alpha} = \{ (i, \alpha^i ) ~:~ 0 \leq i \leq q-2 ~ \}
$$

\begin{theorem}
The set $A_{q,\alpha}$ is an optimal Sidon sequence.
\end{theorem}
\begin{proof}
We have to prove that given four integers $i_1$, $i_2$, $i_3$, and
$i_4$, $0 \leq i_1 ,i_2 ,i_3 ,i_4 \leq q-2$, such that $i_1 \neq
i_3$ and $i_2 \neq i_3$ then the two pairs $(i_1 +i_2 ,
\alpha^{i_1} + \alpha^{i_2})$ and $(i_3 +i_4 , \alpha^{i_3} +
\alpha^{i_4})$ are not equal. Assume the contrary, that for four
such integers we have \vspace{-0.19cm}
\begin{equation}
\label{eq:base} i_1 + i_2 \equiv i_3 + i_4 ~(\text{mod}~q-1),~~~~
\alpha^{i_1} + \alpha^{i_2} = \alpha^{i_3} +\alpha^{i_4} ~.
\vspace{-0.1cm}
\end{equation}
Let $t \equiv i_1 - i_3 \equiv i_4 - i_2 ~(\text{mod}~q-1)$, where
clearly we can assume that $0 < t < q-1$. Hence, we replace
(\ref{eq:base}) with the equations \vspace{-0.19cm}
\begin{equation}
\label{eq:newbase} i_1 - i_3 \equiv i_4 - i_2
~(\text{mod}~q-1),~~~~ \alpha^{i_1} - \alpha^{i_3} = \alpha^{i_4}
-\alpha^{i_2} \vspace{-0.1cm}
\end{equation}
We can substitute $i_1 \equiv i_3 +t ~(\text{mod}~q-1)$ and $i_4
\equiv i_2 +t ~(\text{mod}~q-1)$ in (\ref{eq:newbase}) to obtain
the equation \vspace{-0.19cm}
\begin{equation*}
\alpha^{t+i_3} - \alpha^{i_3} = \alpha^{t+i_2} - \alpha^{i_2},
\vspace{-0.1cm}
\end{equation*}
which is equivalent to the equation \vspace{-0.19cm}
\begin{equation}
\label{eq:last} \alpha^{i_3} (\alpha^t -1) = \alpha^{i_2}
(\alpha^t -1) ~. \vspace{-0.1cm}
\end{equation}
Since $0 < t < q-1$, it follows that $\alpha^t -1 \neq 0$ and
hence from (\ref{eq:last}) we have that $\alpha^{i_3} =
\alpha^{i_2}$, i.e., $i_3 = i_2$ which contradicts the original
choice of $i_2$ and $i_3$. Therefore, $A_{q,\alpha}$ is a Sidon
sequence. The optimality of the sequence is a straight forward
enumeration.
\end{proof}

\begin{remark}
The new construction of Sidon sequences does not help in constructing new
doubly periodic $\cS$-DDC since the abelian group is not $\Z_n$.
\end{remark}

\section{conclusions and Future Research}
\label{sec:conclude}

Sidon sets have many applications in coding theory and in
communication problems. Constructions of optimal Sidon sequences
are rare. We defined periodicity and cyclic arrays in the
two-dimensional case. We proved that unfolded optimal doubly
periodic two-dimensional synchronization patterns are optimal
Sidon sequences. Thus, forming some equivalence between the two
structures. All the results concerning two-dimensional arrays are
generalized readily to higher dimensions. We presented a new
construction of optimal Sidon sequences, with $q-1$ elements, over
an alphabet with $q(q-1)$ elements, where $q$ is a power of a
prime. The main problem for future research in this direction is
to find new constructions for optimal doubly periodic DDCs and new
constructions for optimal Sidon sequences.

\section*{Acknowledgment}
This work was supported in part by the United States-Israel
Binational Science Foundation (BSF), Jerusalem, Israel, under
Grant No. 2006097.


\end{document}